\documentclass[11pt]{article}
\usepackage{amsfonts}
\usepackage{amsmath}
\usepackage{amssymb}
\usepackage{graphicx}
\usepackage[colorlinks=true,linkcolor=blue,citecolor=red,plainpages=false,pdfpagelabels]%
{hyperref}
\usepackage{fullpage}
\usepackage{color}
\usepackage{tabu}
\usepackage{verbatim}
\providecommand{\U}[1]{\protect\rule{.1in}{.1in}}
\newtheorem{theorem}{Theorem}

\newtheorem{corollary}[theorem]{Corollary}

\newtheorem{definition}[theorem]{Definition}

\newtheorem{lemma}{Lemma}

\newtheorem{proposition}[theorem]{Proposition}
\newtheorem{remark}[theorem]{Remark}

\newenvironment{proof}[1][Proof]{\noindent\textbf{#1.} }{\ \rule{0.5em}{0.5em}}

\newcommand{\brac}[1]{\lbrace #1 \rbrace}
\newcommand{\ket}[1]{| #1 \rangle}
\newcommand{\bra}[1]{\langle #1 |}
\newcommand{\mi}{{\min}}
\newcommand{\ma}{{\max}}
\def\U{\mathcal{U}}
\def\A{\mathcal{A}}
\def\B{\mathcal{B}}

\def\D{\mathcal{D}}
\def\E{\mathcal{E}}

\def\H{\mathcal{H}}

\def\M{\mathcal{M}}
\def\N{\mathcal{N}}

\def\P{\mathcal{P}}

\def\T{\mathcal{T}}

\newcommand{\tr}{\operatorname{Tr}}
\newcommand{\unit}{1\!\!1}
\begin{document}

\title{\textbf{Entanglement-assisted private communication over quantum broadcast channels}}
	
\author{Haoyu Qi\thanks{Hearne Institute for Theoretical Physics, Department of Physics and Astronomy,
		Louisiana State University, Baton Rouge, Louisiana 70803, USA. }
\and Kunal Sharma\footnotemark[1]
\and Mark M. Wilde\footnotemark[1]\,
\thanks{Center for Computation and Technology, Louisiana State University, Baton
	Rouge, Louisiana 70803, USA.}
}

\maketitle
\begin{abstract}
	We consider entanglement-assisted (EA) private communication over a quantum broadcast channel, in which there is a single sender and multiple receivers. We divide the receivers into two sets: the decoding set and the malicious set. The decoding set and the malicious set can either be disjoint or can have a finite intersection. For simplicity, we say that a single party Bob has access to the  decoding set and another party Eve has access to the malicious set, and both Eve and Bob have access to the pre-shared entanglement with Alice. The goal of the task is for Alice to communicate classical information reliably to Bob and securely against Eve, and Bob can take advantage of pre-shared entanglement with Alice. In this framework, we establish a lower bound on the one-shot EA private capacity. When there exists a quantum channel mapping the state of the decoding set  to the state of the malicious set,  such  a broadcast channel is said to be degraded. We establish an upper bound on the one-shot EA private capacity in terms of smoothed min- and max-entropies for such channels.  In the limit of a large number of independent channel uses, we prove that the EA private capacity of a degraded quantum broadcast channel is given by a single-letter formula. Finally, we consider two specific examples of degraded broadcast channels and find their capacities. In the first example, we consider the scenario in which one part of Bob's laboratory is compromised by Eve. We show that the capacity for this protocol is given by the conditional quantum mutual information of a quantum broadcast channel, and so we thus provide an operational interpretation to the dynamic counterpart  of the conditional quantum mutual information. In the second example, Eve and Bob have access to mutually exclusive sets of outputs of a broadcast channel.	
	
\end{abstract}

\section{Introduction}
Among the many results of classical information theory, transmitting private information over wiretap channels is  of both conceptual profoundness and practical  relevance \cite{wyner1975wire}. A wiretap channel is modeled as a conditional probability distribution $p_{Y,Z|X}$, in which $X$ models the information a sender Alice intends to transmit, $Y$ models the outcome obtained by a receiver Bob, and $Z$ models what a malicious third-party  Eve holds. The goal of private communication is for Alice to reliably transmit a given message to Bob, while Eve gets negligible information about the transmitted message.

Private communication in quantum information theory is naturally defined by allowing each party to possess a quantum system, as well as a quantum channel to connect Alice to Bob and Eve.  However, in the quantum setting, it is typical to give Eve full control of the environment of the channel from Alice to Bob \cite{ieee2005dev}. This strongest form of security in the quantum setting is guaranteed by the peculiar nature of quantum mechanics, in the form of the no-cloning theorem and the observer effect. Actually, it is the well-known BB84 quantum key distribution protocol \cite{BB84}, a particular kind of private communication protocol, that played a role in the unification of quantum mechanics and classical Shannon theory, which eventually resulted in the birth of what we call quantum Shannon theory today.

The possibility of exploiting shared quantum entanglement prior to communication has been considered extensively in quantum Shannon theory. The superdense coding protocol \cite{bennett1992communication} was the first example to reveal the power of entanglement in the context of communication, in which, by using one ebit and a noiseless quantum channel, one can transmit two bits of classical information. Entanglement-assisted (EA) classical communication over a quantum channel was thereafter one of the problems considered and solved early on  \cite{bennett1999entanglement,bennett2002entanglement,holevo2002entanglement}. Surprisingly, the use of pre-shared entanglement simplifies the problem of determining capacity, in the sense that the optimal rate is given by a single-letter formula: the quantum mutual information of a quantum channel \cite{bennett1999entanglement, bennett2002entanglement,holevo2002entanglement}. Later on, various EA protocols have been studied, including quantum communication \cite{DHW03,devetak2008resource} and  classical communication over quantum broadcast \cite{yard2011quantum,dupuis2010father,wang2016hadamard} and  multiple-access channels \cite{hsieh2008entanglement,qi2017applications}. However, EA private communication has not been considered to the best of our knowledge, although it is practically meaningful and mathematically well-defined. In this work, we consider a general EA private communication protocol over a single-sender multiple-receiver quantum broadcast channel.

The capacity of a channel is an asymptotic concept, defined in the  limit of a large number of channel uses. This notion, which in many cases is given by a simple formula and invokes powerful tools such as  typicality, is one of Shannon's great contributions \cite{shannon1948}. In an effort to bring this notion closer to practice, recently many works have been devoted to the so-called one-shot theory \cite{renner2008security,datta2013a,datta2013b,matthews2014finite}, which studies the maximum amount of information that can be transmitted over a single use of a quantum channel, subject to the error probability being below a certain threshold. Results in one-shot theory typically not only reduce to correct bounds on the capacity in the independent and identically distributed (i.i.d.) limit, but they are also the foundation for further study of correlated quantum channels \cite{bowen2006beyond,RevModPhys.86.1203} and second-order asymptotics \cite{tomamichel2013hierarchy, li2014second,tomamichel2013second,datta2015second,LD16,DTW14, BDL15,TBR15,datta2016second,wilde2017converse,L16}. 

In this work, we consider a general setting for EA private communication, in which a sender and receivers are connected by a quantum broadcast channel $\N_{A\rightarrow \B\cup\E}$. Here $\B$, called the decoding set, includes the systems that Bob holds, and $\E$, called the malicious set, includes all the systems held by Eve. The sets $\B$ and $\E$ need not be disjoint in our model. An $(M,\varepsilon,\delta)$ EA private code is then defined as a set of encoding and decoding channels, such that $M$ transmitted messages can be decoded by Bob with an error probability no more than $\varepsilon \in[0,1]$, and meanwhile the leakage of information to  Eve (defined in what follows) is no more than $\delta\in[0,1]$. The $\varepsilon$-$\delta$-one-shot EA private capacity, denoted as $C_{\operatorname{EP}}^{\varepsilon,\delta}(\N)$, is the largest number $\log_2 M$ such that there exists an $(M,\varepsilon,\delta)$ code for the channel~$\N$. 

Our first result in Theorem~\ref{thm: one-shot-lower-bound} is the following lower bound on the one-shot EA private capacity for $\varepsilon,\delta \in (0,1]$:
\begin{align}
C_{\operatorname{EP}}^{\varepsilon,\delta}(\N)\geq \sup_{\rho_{RA},
\eta_1\in(0,\varepsilon),
\eta_2\in(0,\delta)
}\left[ I_H^{\varepsilon-\eta_1}(R;\B)_\omega-\tilde{I}_{\max}^{\delta-\eta_2}(R;\E)_\omega-\log_2(4\varepsilon/{\eta_1^2})-2\log_2(1/{\eta_2})\right]~,
\label{eq:lower-bnd-main-result}
\end{align}
where  $\rho_{RA}$ is an arbitrary quantum state and $\omega_{R\B\cup\E} =\N_{A\rightarrow \B\cup\E}(\rho_{RA})$. The first two terms in the above expression are the difference of the hypothesis-testing- and smoothed-max-mutual information, which we formally define in Section~\ref{sec:pre}.

To establish the lower bound in \eqref{eq:lower-bnd-main-result}, we use two recently developed techniques: position-based coding \cite{anshu2017one} and convex splitting \cite{anshu2014quantum}. Position-based coding relates the decoding procedure of the receiver to quantum hypothesis testing, while convex splitting works like a one-shot version of the covering lemma (see, e.g.,  \cite[Chapter~17]{W17} for a discussion of the covering lemma). Also, see \cite{qi2017applications} for further developments on the connection between decoding and hypothesis testing in network quantum information theory. These two techniques have been applied in various settings, including EA classical communication over point-to-point quantum channels and broadcast channels \cite{anshu2017one}, private communication \cite{wilde2017position}, classical communication over quantum multiple-access channels \cite{qi2017applications},  state redistribution \cite{anshu2014quantum}, and the quantum Slepian-Wolf problem \cite{anshu2017generalized}. From one-shot lower bounds, it is straightforward to obtain a lower bound on the second-order coding rate by applying second-order expansions of the hypothesis testing relative entropy \cite{tomamichel2013hierarchy,li2014second,datta2016second}, as done, e.g., in \cite{wilde2017position,qi2017applications}.


Our second result in Theorem~\ref{thm:one-shot-upper-bound} is the following upper bound on the one-shot EA private capacity of a quantum broadcast channel:
\begin{align}
C_{\operatorname{EP}}^{\varepsilon,\delta}(\N) \leq \sup_{\rho_{MRA}}[H_{\operatorname{min}}^{\sqrt{2 \delta}}(M \vert R \E)_{\omega} - H_{\operatorname{max}}^{\sqrt{2 \varepsilon}}(M \vert R \B)_{\omega} ] ~,
\end{align}
where $\rho_{MRA}$ is classical on $M$ and quantum on $RA$, and $\omega_{MR\B\cup\E}=\N_{A\rightarrow \B\cup\E}(\rho_{MRA})$.
The definition of smoothed min- and max-conditional entropies are given in Section~\ref{sec:pre}. Theorem~\ref{thm:one-shot-upper-bound-degraded} presents a different one-shot bound in the case that the broadcast channel is degraded (see Definition \ref{def:degraded-broadcast}).

Next, we define the EA private information of a quantum broadcast channel (see \eqref{eq:EA-private-information}), and we prove that it is additive if the quantum broadcast channel is degraded. 
Finally, we prove that the EA private capacity of a degraded broadcast channel is given by the EA private information of the channel. 

We also consider two special cases of degraded quantum broadcast channels. We briefly summarize our results on the first one here, since it gives an operational meaning to the conditional  mutual information (CMI) of a quantum broadcast channel, extending the recent development in \cite{sharma2017conditional}.  The first scenario consists of one part of Bob's laboratory being compromised by Eve, which can be modeled by a broadcast channel $\N_{A\rightarrow  BE}$ in which Bob's laboratory consists of systems $BE$. Bob has access to both systems while Eve has access to system $E$. In this case, the broadcast channel is degraded since $\tr_{B}\brac{\omega_{RBE}}=\omega_{RE}$, where
$\omega_{RBE}=\N_{A\rightarrow  BE}(\rho_{RA})$ and
$\rho_{RA}$ is a bipartite state.
 We prove a single-letter capacity formula for this task, which is given by the CMI of the broadcast channel:
\begin{align}\label{eq:CMI}
C_{\operatorname{EP}}(\N)=\max_{\phi_{RA}}I(R;B|E)_\omega~.
\end{align}

Table~\ref{tb-1} summarizes how our result on the CMI of a broadcast channel fits into the larger context of prior results in quantum Shannon theory. Optimal rates of communication protocols in quantum Shannon theory are often given by entropic quantities. Or put in another way, different communication protocols give operational meanings to different information quantities. An initial resource can either be static or dynamic. A protocol involving a static resource starts with some initial quantum state and realizes some target state at the end, without using a noisy quantum channel as a resource. On the other hand, a protocol involving a dynamic resource, such as a noisy quantum channel, involves the corruption of information when it is transmitted via this channel. For  protocols involving a dynamic resource, the optimal rate is given by an information function of a quantum channel, which usually involves an optimization over states that are fed into the channel.

\begin{table}
\centering
	\begin{tabu} to \textwidth {|p{1.5cm}|c|c|} 
		\hline
		Entropic quantity & Static setting & Dynamic setting \\ 
		\hline\hline
		$H(A)$ & Schumacher compression \cite{schumacher1995quantum} & N/A \\ 
		\hline
		 & state merging \cite{horodecki2007quantum} & quantum communication \cite{ieee2005dev}\\ 
		   $H(A\vert B)$         & entanglement distillation \cite{DW05} & entanglement transmission/generation\\
		\hline
		 & quantum one-time pad \cite{schumacher2006quantum}& \\
		$I(A;B)$& erasure of quantum correlation \cite{GPW05} & EA communication \cite{bennett1999entanglement}\\
		& quantum Slepian-Wolf \cite{horodecki2007quantum}&\\
		\hline
		 & state redistribution \cite{DY08} & {\textbf{EA private communication}}\\
		         $I(A;B|C)$         & state deconstruction \cite{berta2016deconstruction,BBMW18}& {\textbf{ over quantum broadcast channel}}\\
		                  & conditional quantum one-time pad \cite{sharma2017conditional}& (\textbf{this paper})\\
		\hline
	\end{tabu}
	\caption{Entropic quantities and the corresponding static and dynamic settings. In this work, we establish an operational meaning for the CMI of a quantum broadcast channel $\N_{A\rightarrow BE}$: it is the optimal rate of EA private communication over that  channel (see Section~\ref{sec:ex} for details). }
	\label{tb-1}
\end{table}


The rest of our paper is organized as follows. In Section \ref{sec:pre}, we summarize definitions and lemmas relevant to our proofs. We consider bounds on the one-shot EA private capacity in Section \ref{sec:capacity-degrade}. There we establish both lower and upper bounds on the one-shot EA private capacity of an arbitrary quantum broadcast channel. By combining these results, we arrive at a single-letter formula for the EA private capacity of a degraded quantum broadcast channel in the asymptotic setting. In Section \ref{sec:ex}, we consider two special cases of a two-receiver broadcast channel. As corollaries of our main theorem, we establish EA private capacities for both cases. In the first scenario, we prove that the $\operatorname{CMI}$ of a quantum broadcast channel is the optimal rate. Finally, we summarize our main results and discuss  future directions in Section~\ref{sec:conclusion}.

\section{Preliminaries}

\label{sec:pre}

 We use notation and concepts that are standard in quantum information theory
and point readers to \cite{W17} for background. In the rest of this
section, we review concepts that are less standard and set some notation that
will be used later in the paper.

\bigskip

\noindent\textbf{Trace distance, fidelity, and purified distance.} Let $\mathcal{D}%
(\mathcal{H})$ denote the set of density operators acting on a Hilbert space
$\mathcal{H}$ and  $\mathcal{D}_{\leq}(\mathcal{H})$ the set of
subnormalized density operators (with trace not exceeding one) acting on
$\mathcal{H}$.
The trace distance
between two quantum states $\rho,\sigma\in\mathcal{D}(\mathcal{H})$\ is equal
to $\left\Vert \rho-\sigma\right\Vert _{1}$, where $\left\Vert C\right\Vert_{1}\equiv\operatorname{Tr}\{\sqrt{C^{\dag}C}\}$ for any operator $C$. It has
a direct operational interpretation in terms of the distinguishability of
these states. 
The fidelity between two quantum states is
defined as $F(\rho,\sigma)\equiv\left\Vert \sqrt{\rho}\sqrt{\sigma}\right\Vert
_{1}^{2}$ \cite{U76},
which is
invariant with respect to isometries and monotone non-decreasing with respect
to channels. The sine distance or $C$-distance between two quantum states
$\rho,\sigma\in\mathcal{D}(\mathcal{H})$ is defined as
\begin{equation}
C(\rho,\sigma)\equiv\sqrt{1-F(\rho,\sigma)},
\end{equation}
and it was proven to be a metric in \cite{rastegin2002relative,rastegin2003lower,rastegin2006sine,gilchrist2005distance}. It was
later~\cite{tomamichel2009fully} (under the name \textquotedblleft purified
distance\textquotedblright) shown to be a metric on subnormalized states
$\rho,\sigma\in\mathcal{D}_{\leq}(\mathcal{H})$ via the embedding
\begin{equation}
P(\rho,\sigma)\equiv C(\rho\oplus\left[  1-\operatorname{Tr}\{\rho\}\right]
,\sigma\oplus\left[  1-\operatorname{Tr}\{\sigma\}\right]
)\,.\label{eq:purified-distance}%
\end{equation}
The following inequality relates trace distance and purified distance:%
\begin{equation}
\frac{1}{2}\left\Vert \rho
\oplus\left[  1-\operatorname{Tr}\{\rho\}\right]
-\sigma
\oplus\left[  1-\operatorname{Tr}\{\sigma\}\right]
\right\Vert _{1}\leq P(\rho,\sigma
).\label{eq:TD-to-PD}%
\end{equation}
For a state $\rho\in D(\H)$, we define the ball of $\varepsilon$-close subnormalized states around $\rho$ as
\begin{align}
\B^\varepsilon(\rho) = \brac{\bar{\rho}\in \D_{\leq}(\H) : P(\bar{\rho},\rho)\leq\varepsilon}~.
\end{align}

\bigskip

\noindent\textbf{Relative entropies and variances.} The quantum relative entropy of two
states $\omega$ and $\tau$ is defined as \cite{U62}%
\begin{equation}
D(\omega\Vert\tau)\equiv\operatorname{Tr}\{\omega\lbrack\log_{2}\omega
-\log_{2}\tau]\}
\end{equation}
whenever $\operatorname{supp}(\omega)\subseteq\operatorname{supp}(\tau)$, and
it is equal to $+\infty$ otherwise. 

The hypothesis testing relative entropy \cite{BD10,wang2012one}\ of states $\omega$ and
$\tau$ is defined as%
\begin{equation}
D_{H}^{\varepsilon}(\omega\Vert\tau)\equiv-\log_{2}\inf_{\Lambda}\left\{
\operatorname{Tr}\{\Lambda\tau\}:0\leq\Lambda\leq I\wedge\operatorname{Tr}%
\{\Lambda\omega\}\geq1-\varepsilon\right\}  .
\end{equation}
The max- and min-relative entropy for states $\omega$ and $\tau$ are defined as
\cite{datta2009min,konig2009operational}%
\begin{align}
D_{\max}(\omega\Vert\tau) &\equiv\inf\left\{  \lambda\in\mathbb{R}:\omega
\leq2^{\lambda}\tau\right\}~,\\
D_{\min}(\omega\Vert\tau)
&\equiv-\log_2 F(\omega,\tau)~.
\end{align}
The following relation between the min- and max-relative entropies holds \cite[Theorem~7]{muller2013quantum}
\begin{align}\label{ieq: Dmax>Dmin}
D_{\max}(\omega\Vert\tau)\geq D_{\min}(\omega\Vert\tau)~.
\end{align}
The smoothed max- and min-relative entropy for states $\omega$ and $\tau$, and a parameter
$\varepsilon\in(0,1)$ are defined as \cite{datta2009min,konig2009operational}%
\begin{align}
D_{\max}^{\varepsilon}(\omega\Vert\tau)&\equiv\inf_{\bar{\omega}\in\B^\varepsilon(\omega)} D_{\max}(\bar{\omega}\Vert\tau)~,\\
D_{\min}^{\varepsilon}(\omega\Vert\tau)&\equiv\sup_{\bar{\omega}\in\B^\varepsilon(\omega)} D_{\min}(\bar{\omega}\Vert\tau)~.
\end{align}

\bigskip

\noindent\textbf{Conditional entropies and mutual informations.} 
Conditional entropies play an important role in our converse proof. The max- and min-conditional entropies are defined as \cite{renner2008security}
\begin{align}
\label{eq:Hmax-def}
H_{\max}(A\vert B)_\rho & \equiv -\inf_{\sigma_B\in\D(\H_B)} D_{\min}(\rho_{AB}\Vert \unit_A\otimes\sigma_B)~,\\
\label{eq:Hmin-def}
H_{\min}(A\vert B)_\rho & \equiv -\inf_{\sigma_B\in\D(\H_B)} D_{\max}(\rho_{AB}\Vert \unit_A\otimes\sigma_B)~,
\end{align}
along with their smoothed versions:
\begin{align}
H^\varepsilon_{\max}(A\vert B)_\rho &\equiv \inf_{\bar{\rho}\in\B^\varepsilon(\rho)}H_{\max}(A\vert B)_{\bar{\rho}}~,\\
H^\varepsilon_{\min}(A\vert B)_\rho &\equiv \sup_{\bar{\rho}\in\B^\varepsilon(\rho)}H_{\min}(A\vert B)_{\bar{\rho}}~.
\end{align}
If the $B$ system is trivial, the conditional entropies reduce to max- and min-entropies:
\begin{align}
H_\ma(A)_\rho &= \log_2\left\Vert\sqrt{\rho_A}\right\Vert_1^2~,\\
H_\mi(A)_\rho& = -\log_2\lambda_\ma(\rho_A)~.
\end{align}

 We can define different one-shot mutual informations by using different relative entropies.  It turns out that the max-mutual information often appears in one-shot bounds of various protocols. 
There are several different ways to define max-mutual information in general \cite{BCR09,ciganovic2014smooth}, but what we employ in the convex-split lemma below is  the following variation \cite{anshu2017one}:
\begin{align}
\tilde{I}^\varepsilon_{\max}(B;A)_\rho\equiv \inf_{\rho'\in\B^\varepsilon(\rho)} D_{\max}(\rho'_{AB}\Vert\rho_A\otimes\rho_B')~.
\end{align}
The $\varepsilon$-hypothesis-testing-mutual information is defined here as
\begin{align}
I_H^\varepsilon(A;B)_\rho\equiv D_H^\varepsilon(\rho_{AB}\Vert \rho_A\otimes\rho_B)~.
\end{align}
\bigskip

\noindent\textbf{Hayashi--Nagaoka operator inequality.} A key tool in analyzing error
probabilities in communication protocols is the Hayashi--Nagaoka operator
inequality  \cite{hayashi2003generalCapacity}:\ given operators $S$ and $T$ such that $0\leq S\leq I$
and $T\geq0$, the following inequality holds for all $c>0$:
\begin{equation}
I-(S+T)^{-1/2}S(S+T)^{-1/2}\leq(1+c)(I-S)+(2+c+c^{-1})T. \label{eq:HN-ineq}%
\end{equation}

\bigskip

\noindent\textbf{Convex-split lemma.} The convex-split lemma from \cite{anshu2014quantum} has been
a key tool used in recent developments in quantum information theory
\cite{anshu2014quantum,anshu2017one}. Here, we state a slight variant of the convex-split lemma from \cite{wilde2017position}, which can be
helpful for obtaining one-shot bounds for privacy and ensuing bounds
on second-order coding rates.

Let $\rho_{AB}$ be a state, and let
$\tau_{A_{1}\cdots A_{K}B}$ be the following state:%
\begin{equation}
\tau_{A_{1}\cdots A_{K}B}\equiv\frac{1}{K}\sum_{k=1}^{K}\rho_{A_{1}}%
\otimes\cdots\otimes\rho_{A_{k-1}}\otimes\rho_{A_{k}B}\otimes\rho_{A_{k+1}%
}\otimes\cdots\otimes\rho_{A_{K}}.
\end{equation}
Let $\varepsilon\in(0,1)$ and $\eta\in(0,\sqrt{\varepsilon})$. If%
\begin{equation}
\log_{2}K=~ \tilde{I}_{\max}^{\sqrt{\varepsilon}-\eta}(B;A)_{\rho}%
+2\log_{2}\!\left(  \frac{1}{\eta}\right)  ,
\end{equation}
then%
\begin{equation}
P(\tau_{A_{1}\cdots A_{K}B},\rho_{A_{1}}\otimes\cdots\otimes\rho_{A_{K}%
}\otimes\widetilde{\rho}_{B})\leq\sqrt{\varepsilon},
\label{eq:perf-error-convex-split}%
\end{equation}
for some state $\widetilde{\rho}_{B}$ such that $P(\rho_{B},\widetilde{\rho
}_{B})\leq\sqrt{\varepsilon}-\eta$.



\section{One-shot bounds for EA private communication over a quantum broadcast channel}

\label{sec:capacity-degrade}

We consider a quantum broadcast channel $\N_{A\rightarrow \B\cup\E}$, for which $\B$ is the set of systems held by Bob, while $\E$ is the set of systems held by Eve. We call $\B$ the \textit{decoding set} and $\E$ the \textit{malicious set}. Notice that we do not assume any relationship between the two sets $\B$ and $\E$. For instance, it is possible that $\B\cap\E \neq\emptyset$ or $\E\subset \B$. It is this freedom that gives our model some generality. 

In a protocol for EA private communication,
a sender Alice would like to transmit a classical message $m$, chosen from a set $\M=\brac{1,\ldots,M}$ where $M \in \mathbb{Z}^+$, to Bob via the quantum broadcast channel $\N_{A\rightarrow \B\cup\E}$.  She and the receivers also pre-share entanglement to assist their communication, represented by some bipartite state $\Psi_{RA'}$. Moreover, we also allow Eve to have access to this pre-shared entanglement. 
The goal of EA private communication is for Bob, who holds systems $R\B$, to reliably decode Alice's transmitted message, while Eve, who holds systems $R\E$, can only get negligible information about Alice's message. Fix
$M \in \mathbb{Z}^+$, $\varepsilon\in[0,1]$, and $\delta \in [0,1]$.
We define an $(M,\varepsilon, \delta)$ code to be a set of encoding channels
$\brac{\E^m_{A'\rightarrow A}}_m$
and a decoding positive operator-valued measure (POVM) $\brac{\Lambda^m_{R\B}}_m$, such that

\begin{enumerate} 
\item the classical messages can be reliably decoded by Bob:
\begin{align}\label{condition:reliability}
\max_{m\in\M}p_e(m)\leq \varepsilon~,
\end{align}
where
$p_e(m) = \tr\brac{(I-\Lambda^m_{R\B})\rho^m_{R\B}}$ and
$
\rho^m_{R\B\cup\E} =\N_{A\rightarrow \B\cup\E}(\E^m_{A'\rightarrow A}(\Psi_{RA'}))
$, and
\item each classical message is $\delta$-secure:
\begin{align}\label{condition:security}
\frac{1}{2}\Vert \rho^m_{R\E}-\sigma_{R\E}\Vert_1\leq\delta,~~~\forall m\in\M~,
\end{align}
where $\sigma_{R\E}$ is a fixed state.
\end{enumerate}
In the above, $\sigma_{R\E}$ is some constant state that does not contain any information about Alice's message (one can show that \eqref{condition:security} guarantees that the mutual information $I(M;R\E)$ is small \cite[Section~23.1.1]{W17}). For fixed $\varepsilon, \delta$, let $C_{\operatorname{EP}}^{\varepsilon,\delta}(\N_{A\rightarrow \B\cup\E})$  denote the one-shot EA private capacity, i.e., the largest value of $\log_2 M$ for which there exists an $(M, \varepsilon, \delta)$ code. The EA private capacity of the quantum broadcast channel $\N_{A\rightarrow \B\cup\E}$ is defined as
\begin{align}
C =\lim_{\varepsilon,\delta\rightarrow 0}\liminf_{n\rightarrow\infty} \frac{1}{n} C_{\operatorname{EP}}^{\varepsilon,\delta}(\N^{\otimes n}_{A\rightarrow \B\cup\E}).
\end{align}

In our paper, we focus our attention mostly on degraded quantum broadcast channels, which are defined as follows:
\begin{definition}[Degraded broadcast channel]
	\label{def:degraded-broadcast}
	Let $\N_{A\rightarrow \B\cup\E}$ be a quantum broadcast channel with a decoding set $\B$ and a malicious set $\E$. The channel $\N_{A\rightarrow \B\cup\E}$ is degraded if there exists a quantum channel $\T:\D(\H_{\B})\rightarrow\D(\H_{\E})$ such that
	\begin{align}
	\T(\tr_{\B\cup\E\setminus\B}\brac{\N(\rho)}) = \tr_{\B\cup\E\setminus\E}\brac{\N(\rho)}
	\end{align}
	for all states $\rho\in\D(\H_A)$. Here $\A\setminus\B$ includes all the systems in $\A$ except those in $\B$.
\end{definition}

\subsection{Lower bound on the one-shot EA private capacity} \label{subsec:oneshot}

In this section, we construct a  code for EA private communication based on the techniques of position-based coding \cite{anshu2017one} and convex splitting \cite{anshu2014quantum}. 
\begin{theorem}\label{thm: one-shot-lower-bound}
	Let $\N_{A\rightarrow \B\cup\E}$
	be a quantum broadcast channel
	with decoding set $\B$ and malicious set~$\E$. For fixed $\varepsilon \in (0,1)$ and $\delta\in (0,1)$, the one-shot EA
	private capacity  is bounded from below as
	\begin{align}
	C_{\operatorname{EP}}^{\varepsilon,\delta}(\N)\geq \sup_{\rho_{RA},
	\eta_1\in (0,\varepsilon), \eta_2\in (0,\sqrt{\delta})}I_H^{\varepsilon-\eta_1}(R;\B)_\omega-\tilde{I}_{\max}^{\sqrt{\delta}-\eta_2}(\E;R)_\omega-\log_2\!\left(\frac{4\varepsilon}{\eta_1^2}\right)-2\log_2\!\left(\frac{1}{\eta_2}\right)~.
	\end{align}
In the above, $\rho_{RA}$ is an arbitrary  quantum state and $\omega_{R\B\cup\E}=\N_{A\rightarrow \B\cup\E}(\rho_{RA})$.
\end{theorem}
\begin{proof} The proof is related to the approach given in \cite{wilde2017position}, which more generally is inspired by the well known approach from \cite{wyner1975wire}.
\\
\newline
	\noindent\textit{Encoding:} Alice and Bob prepare $M$ blocks of entangled states, each of which is the tensor-product of $K$  bipartite states $\rho_{RA}$. That is, we take the pre-shared entangled state before the communication begins to be
	\begin{align}
	\rho_{R^{MK}A^{MK}} = \rho_{R_{(1,1)}A_{(1,1)}}\otimes \rho_{R_{(1,2)}A_{(1,2)}}\otimes \cdots \otimes\rho_{R_{(1,K)}A_{(1,K)}}\otimes \cdots \otimes\rho_{R_{(M,K)}A_{(M,K)}}~.
	\end{align}
	To send message $m$, Alice first chooses a local key variable $k$ uniformly at random and then sends the $(m,k)$th $A$ system through the quantum channel $\N$. Therefore, after the transmission, the state for Bob and Eve is as follows:
	\begin{align}
	\rho^{m,k}_{R^{MK}\B\cup\E}=\rho_{R_{(1,1)}}\otimes \cdots \otimes\rho_{R_{(m,k-1)}}\otimes\omega_{R_{(m,k)}\B\cup \E}\otimes \cdots \otimes\rho_{R_{(M,K)}}~,
	\end{align}
	where $\omega_{R\B\cup\E}=\N_{A\rightarrow \B\cup\E}(\rho_{RA})$.\\
	\newline
	\noindent
	\textit{Reliable decoding:}  From previous work \cite{anshu2017one,wilde2017position,qi2017applications}, we know that as long as 
	\begin{align}\label{eq:MK}
	\log_2 MK = I_H^{\varepsilon-\eta_1}(R;\B)_\omega -\log_2(4\varepsilon/\eta_1^2)~,
	\end{align}
	where $\varepsilon\in (0,1)$ and $\eta_1\in (0,\varepsilon)$, we have the following bound
	\begin{align}
	\tr\brac{(I-\Lambda^{m,k}_{R\B})\rho^{m,k}_{R\B}}\leq \varepsilon~,~~~\forall~m,k~.
	\end{align}
	where $\{\Lambda^{m,k}_{R\B}\}$ is a POVM  built from the test operator for the hypothesis testing relative entropy $D_H^{\varepsilon-\eta_1}(\omega_{R\B} \Vert \omega_R\otimes \omega_B)$ that optimally distinguishes between $\omega_{R\B}$ and $\omega_R\otimes \omega_B$. In particular, see \cite[Theorem~8]{qi2017applications} for more details.\\
	\newline
	\noindent\textit{Security:} Since for each message $m$, the local key $k$ is chosen uniformly at random, the state held by the malicious party is as follows:
	\begin{align}
	\rho^m_{R^{MK}\E}=\frac{1}{K}\sum_{k=1}^K\rho^{m,k}_{R^{MK}\E}~.
	\end{align}
	Now invoking the convex-split lemma (recalled at the end of Section~\ref{sec:pre}), as long as
	\begin{align}\label{eq:key}
	\log_2 K =\tilde{I}_{\max}^{\sqrt{\delta}-\eta_2}(\E;R)_\rho + 2\log_2({1}/{\eta_2})~,
	\end{align}
	where $\eta_2\in (0,\sqrt{\delta})$, we have the following bound for the trace distance:
	\begin{align}
	&\frac{1}{2}\left\Vert \rho ^m_{R^{MK}\E} - \rho_{R^{MK}}\otimes  \tilde{\rho}_{\E}\right\Vert_1 \nonumber\\
	&= \frac{1}{2}\left\Vert \frac{1}{K}\sum_{k=1}^K \rho_{R_{(m,1)}}\otimes\cdots\otimes\omega_{R_{(m,k)}\E}\otimes \cdots \otimes \rho_{R_{(m,K)}}-\rho_{R^K}\otimes\tilde{\rho}_{\E}\right\Vert_1\\
	&\leq  P\!\left(\frac{1}{K}\sum_{k=1}^K \rho_{R_{(m,1)}}\otimes\cdots\otimes\omega_{R_{(m,k)}\E}\otimes \cdots \otimes \rho_{R_{m,K}},\rho_{R^K}\otimes\widetilde{\rho}_{\E}\right)\\
	&\leq  \sqrt{\delta}~,
	\end{align}
	where $\widetilde{\rho}_{\E}$ is a state such that $P(\rho_{\E},\widetilde{\rho}_{\E})\leq\sqrt{\delta}-\eta_2$, and $ \eta_2\in (0,\sqrt{\delta})$.
	The first equality follows from the property $\Vert\sigma\otimes \tau-\omega\otimes\tau\Vert_1=\Vert\sigma-\omega\Vert_1$. The first inequality follows from the definition of purified distance. The last inequality is due to the convex-split lemma and the choice in \eqref{eq:key}.
	
	Therefore, by combining \eqref{eq:MK} and \eqref{eq:key}, we have an $(M,\varepsilon, \delta)$ code with 
	\begin{align}
	\log_2 M &= I_H^{\varepsilon-\eta_1}(R;\B)_\omega-\tilde{I}_{\max}^{\sqrt{\delta}-\eta_2}(\E;R)_\omega-\log_2\!\left(\frac{4\varepsilon}{\eta_1^2}\right)-2\log_2\!\left(\frac{1}{\eta_2}\right)\\
	& = \sup_{\rho_{RA},
	\eta_1\in (0,\varepsilon), \eta_2\in (0,\sqrt{\delta})} [I_H^{\varepsilon-\eta_1}(R;\B)_\omega-\tilde{I}_{\max}^{\sqrt{\delta}-\eta_2}(\E;R)_\omega-\log_2\!\left(\frac{4\varepsilon}{\eta_1^2}\right)-2\log_2\!\left(\frac{1}{\eta_2}\right)]~.
	\end{align}
	The last equality follows because the first equality holds for any input state $\rho_{RA}$, and for any value of $\eta_1 \in (0, \varepsilon)$ and $\eta_2 \in (0, \sqrt{\delta})$.
	Since the one-shot capacity is defined to be the largest value of $\log_2 M$ for which there exists an $(\M,\varepsilon,\delta)$ code, the desired result follows.
\end{proof}

\bigskip

\textbf{Lower bound on the second-order coding rate.} Defining the relative
entropy variance of two states $\omega$ and $\tau$ as \cite{tomamichel2013hierarchy,li2014second}%
\begin{equation}
V(\omega\Vert\tau)=\operatorname{Tr}\{\omega\left[  \log_{2}\omega-\log
_{2}\tau-D(\omega\Vert\tau)\right]  ^{2}\},
\end{equation}
and the inverse cumulative Gaussian distribution function as $\Phi
^{-1}(\varepsilon)\equiv\sup\left\{  a\in\mathbb{R}\ |\ \Phi(a)\leq
\varepsilon\right\}  $, where%
\begin{equation}
\Phi(a)\equiv\frac{1}{\sqrt{2\pi}}\int_{-\infty}^{a}dx\ \exp(-x^{2}/2),
\end{equation}
we can obtain a lower bound on the second-order coding rate for EA\ private
communication, in a way similar to what was reported in \cite{wilde2017position}. Indeed, recall the following second-order expansions \cite{tomamichel2013hierarchy,li2014second,datta2016second}:%
\begin{align}
D_{H}^{\varepsilon}(\omega^{\otimes n}\Vert\tau^{\otimes n})  & =nD(\omega
\Vert\tau)+\sqrt{nV(\omega\Vert\tau)}\Phi^{-1}(\varepsilon)+O(\log n),\\
D_{\max}^{\sqrt{\varepsilon}}(\omega^{\otimes n}\Vert\tau^{\otimes n})  &
=nD(\omega\Vert\tau)-\sqrt{nV(\omega\Vert\tau)}\Phi^{-1}(\varepsilon)+O(\log
n).
\end{align}
Let us define the mutual information variance $V(A;B)_{\rho}$\ of a bipartite
state $\rho_{AB}$ as%
\begin{equation}
V(A;B)_{\rho}\equiv V(\rho_{AB}\Vert\rho_{A}\otimes\rho_{B}).
\end{equation}
Then by taking $\eta_{1}=\eta_{2}=1/\sqrt{n}$, and applying the above
expansions, as well as \cite[Lemma~1]{wilde2017position}, we find the following lower bound on the
second-order coding rate for EA\ private communication over the broadcast
channel $\mathcal{N}$:%
\begin{multline}
\label{eq:2nd-order-lower-bnd}
C_{\operatorname{EP}}^{\varepsilon,\delta}(\mathcal{N}^{\otimes n})\geq n\left[
I(R;\mathcal{B})_{\omega}-I(R;\mathcal{E})_{\omega}\right]  \\
+\sqrt{nV(R;\mathcal{B})_{\omega}}\Phi^{-1}(\varepsilon)+\sqrt
{nV(R;\mathcal{E})_{\omega}}\Phi^{-1}(\delta)+O(\log n),
\end{multline}
for $\varepsilon,\delta\in(0,1)$ and $\omega_{R\B\cup\E}=\N_{A\rightarrow \B\cup\E}(\rho_{RA})$ for some state $\rho_{RA}$.

\subsection{Upper bound on the one-shot EA private capacity}

In the proof of Theorem~\ref{thm:one-shot-upper-bound} below, we derive an  upper bound on the one-shot EA private capacity of a quantum broadcast channel.
This upper bound  coincides with the lower bound from Theorem~\ref{thm: one-shot-lower-bound} in the asymptotic, i.i.d.~limit.
\begin{theorem}
\label{thm:one-shot-upper-bound}
	Let $\N_{A\rightarrow \B\cup\E}$ be a quantum broadcast channel with a decoding set $\B$ and a malicious set $\E$, and let $\varepsilon$, $\delta \in (0,1)$. Then the one-shot EA private capacity is bounded from above as 
	\begin{align} \label{eq:one-shot-upper-bound-broadcast}
	C_{\operatorname{EP}}^{\varepsilon,\delta}(\N) \leq \sup_{\rho_{MRA}}[H_{\operatorname{min}}^{\sqrt{2 \delta}}(M \vert R \E)_{\omega} - H_{\operatorname{max}}^{\sqrt{2 \varepsilon}}(M \vert R \B)_{\omega} ] ~,
	\end{align}
	where $\rho_{MRA}$ is classical on $M$ and quantum on $RA$ and $\omega_{MR\B\cup\E}=\N_{A\rightarrow \B\cup\E}(\rho_{MRA})$.	
\end{theorem}
\begin{proof}
We begin by establishing an upper bound on $
\log_2 M$ for an arbitrary $(M,\varepsilon,\delta)$ EA private communication code (with the state $\omega$ defined in what follows), essentially by following an approach similar to that in \cite{renes2011noisy}.
To establish the upper bound, we consider the task of EA secret key distribution, which in turn gives an upper bound on the one-shot EA private capacity. In this task, Alice picks a classical message uniformly at random, places it in a system $M$, and makes a copy of it in a system $M'$. The goal at the end is to produce a secure and perfectly correlated key between her and Bob, such that Bob has a copy of Alice's message. Therefore, the initial state of Alice's systems is as follows:
\begin{align}
\overline{\Phi}_{MM'} =\sum_{m}\frac{1}{M}\ket{m}\bra{m}_M\otimes \ket{m}\bra{m}_{M'}~.
\end{align}
For an arbitrary $(M,\varepsilon,\delta)$ code, the combined state of Bob and Eve's systems after one use of the broadcast channel $\N_{A\rightarrow \B\cup\E}$ is as follows:
\begin{align}
\omega_{MR\B\cup\E} = \frac{1}{M}\sum_m\ket{m}\bra{m}_M\otimes\rho^m_{R\B\cup\E}~,
\end{align}
where $
\rho^m_{R\B\cup\E} =\N_{A\rightarrow \B\cup\E}(\E^m_{A'\rightarrow A}(\Psi_{RA'}))$.
After the decoding procedure, Alice and Bob end up with imperfect shared randomness, represented by the following state:
\begin{align}
\sigma_{MM'}= \frac{1}{M}\sum_{m,m'}p(m'|m)\ket{m}\bra{m}_M\otimes\ket{m'}\bra{m'}_{M'}~.
\end{align}
Here, $p(m'|m)$ is the probability of Bob decoding $m'$ when the message transmitted by Alice is $m$. 

Next, we find an upper bound on the trace distance between $\sigma_{MM'}$ and $\overline{\Phi}_{MM'}$. Consider the following chain of inequalities:
\begin{align}
& \frac{1}{2}\Vert \sigma_{MM'}-\overline{\Phi}_{MM'}\Vert_1 \nonumber \\
& = \frac{1}{2M}\left\Vert\sum_m\ket{m}\bra{m}_M\otimes\left[\sum_{m'}p(m'|m)\ket{m'}\bra{m'}_{M'}-\ket{m}\bra{m}_{M'}\right]\right\Vert_{1}\\
&= \frac{1}{2M} \sum_m\left\Vert \sum_{m'\neq m}p(m'|m)\ket{m'}\bra{m'}_{M'} + (p(m|m)-1)\ket{m}\bra{m}_{M'}\right\Vert_{1}\\
&= \frac{1}{2M}\sum_m \left(2\sum_{m'\neq m}p(m'|m)\right)\\
&\leq \varepsilon~.
\end{align}
The first equality follows from the direct-sum property of trace norm. The second equality follows from the triangle inequality. To obtain the last inequality, we apply the reliable decoding condition of an $(M,\varepsilon,\delta)$ code. 

We now show that $H_{\max}^{\sqrt{2\varepsilon}}(M|R\B)_\omega\leq 0$. Consider the following chain of inequalities:
\begin{align}\label{ineq:H_max<0}
H_{\max}^{\sqrt{2\varepsilon}}(M|R\B)_\omega \leq  H_{\max}^{\sqrt{2\varepsilon}}(M|M')_\sigma \leq H_{\max}(M|M')_{\overline{\Phi}} \leq 0~.
\end{align}
The first inequality follows from the data processing inequality for the smoothed max-conditional entropy (see, e.g., \cite{tomamichel2012framework}). The second inequality follows from the definition of smoothed-max-conditional entropy.

Next, we show that $H_{\min}^{\sqrt{2\delta}}(M\vert R\E)_\omega\geq \log_2 M$, by invoking the security condition of the code. From the security condition, we know that $\frac{1}{2}\Vert\rho^m_{R\E}-\sigma_{R\E}\Vert\leq \delta$ for all messages $m$, and we thus have
\begin{align}
\frac{1}{2}\Vert \omega_{MR\E} -\omega_M\otimes \sigma_{R\E}\Vert_1\leq \delta~.
\end{align}
Therefore, by using the definition of purified distance and the Powers-Stormer inequality \cite{powers1970free}, we find that $\omega_M\otimes\sigma_{R\E}\in \B^{\sqrt{2\delta}}(\omega_{MR\E})$. The rest is straightforward:
\begin{align}\label{ineq:H_min>logM}
H_{\min}^{\sqrt{2\delta}}(M\vert R\E)_\omega \geq H_{\min}(M\vert R\E)_{\omega_M\otimes \sigma_{R\E}} \geq H_{\min}(M)_\omega = \log_2 M~. 
\end{align}
Using \eqref{ineq:H_max<0} and \eqref{ineq:H_min>logM}, we establish the following upper bound on the amount of transmitted information:
\begin{align}
\log_2 M & \leq H_{\min}^{\sqrt{2\delta}}(M\vert R\E)_\omega-H_{\max}^{\sqrt{2\varepsilon}}(M|R\B)_\omega \label{eq:one-shot-not-deg-without-sup}\\
&\leq \sup_{\rho_{MRA}}[H_{\operatorname{min}}^{\sqrt{2 \delta}}(M \vert R \E)_{\omega} - H_{\operatorname{max}}^{\sqrt{2 \varepsilon}}(M \vert R \B)_{\omega} ]~. \label{eq:one-shot-not-deg}
\end{align}
Since these inequalities hold for any value of $\log_2M$, the desired result in \eqref{eq:one-shot-upper-bound-broadcast} follows. 
\end{proof}

\begin{theorem}
\label{thm:one-shot-upper-bound-degraded}
If $\N_{A\rightarrow \B\cup\E}$ is a degraded quantum broadcast channel, then the one-shot EA private capacity is bounded from above as 
	\begin{align}\label{eq:one-shot-upper-bound-degraded-broadcast}
	C_{\operatorname{EP}}^{\varepsilon,\delta}(\N) \leq \sup_{\rho_{RA}}\big[ H_{\min}^{3\sqrt{2\varepsilon}+2\sqrt{2\delta}}(R|\E)_\omega- H_{\max}^{3\sqrt{2\varepsilon}+2\sqrt{2\delta}}(R|\B)_\omega\big]+ f(\varepsilon,\delta)~,
	\end{align}
	where the optimization is over all bipartite states $\rho_{RA}$. 
	Here, $\omega_{R\B\cup\E}=\N_{A\rightarrow \B\cup\E}(\rho_{RA})$, $\delta,\varepsilon \in(0,\frac{1}{8})$, $ 3\sqrt{2\varepsilon}+2\sqrt{2\delta}<1$, and
	 $f(\varepsilon,\delta) = -\log_2(1-\sqrt{1-8\delta})(1-\sqrt{1-8\varepsilon})$.
\end{theorem}

\begin{proof}
	Suppose that the quantum broadcast channel $\N_{A\rightarrow\B\cup\E}$ is degraded. In what follows, we apply the following chain rules to \eqref{eq:one-shot-not-deg-without-sup} for $\varepsilon,\varepsilon',\varepsilon''\in(0,1)$ 	
	(\cite[Theorem~13]{vitanov2013chain} and its dual):
	\begin{align}
	H_{\min}^{\varepsilon}(A|BC)_\rho \leq H_{\min}^{\varepsilon'+2\varepsilon+\varepsilon''}(AB|C)_\rho - H_{\min}^{\varepsilon'}(B|C)_\rho+g(\varepsilon'')~,
	\label{eq:chain-1}\\
	H_{\max}^{\varepsilon}(A|BC)_\rho \geq H_{\max
	}^{\varepsilon+2\varepsilon'+\varepsilon''}(AB|C)_\rho - H_{\max}^{\varepsilon'}(B|C)_\rho -g(\varepsilon'')~\label{eq:chain-2},
	\end{align}
	where $g(x)=-\log_2(1-\sqrt{1-x^2})$ and $g(x)\sim \log_2({x}^{-1})$ for small $x$. We then have that
	\begin{align}
	\log_2 M &\leq H_{\min}^{\delta'+2\sqrt{2\delta}+\delta''}(MR|\E)_\omega - H_{\min}^{\delta'}(R|\E)_\omega + g({\delta''})\nonumber\\
	& \qquad -\left[H_{\max}^{\sqrt{2\varepsilon}+2\varepsilon'+\varepsilon''}(MR|\B)_\omega - H_{\max}^{\varepsilon'}(R|\B)_\omega -g(\varepsilon'')\right]\\
	&= H_{\min}^{\delta'+2\sqrt{2\delta}+\delta''}(MR|\E)_\omega - H_{\max}^{\sqrt{2\varepsilon}+2\varepsilon'+\varepsilon''}(MR|\B)_\omega -\left[H_{\min}^{\delta'}(R|\E)_\omega-H_{\max}^{\varepsilon'}(R|\B)_\omega\right] \nonumber \\
	&\qquad +g(\delta'')+g(\varepsilon'')~, \label{ineq: 1}
	\end{align}
	for $\varepsilon,\varepsilon',\varepsilon'',\delta,\delta',\delta''\in(0,1)$.

To proceed from here, we invoke Lemma~9 from \cite{morgan2014pretty}:
\begin{align}\label{ineq:from-pretty-converse}
H_{\min}^\varepsilon(A\vert B)_\rho \geq H_{\max}^{\sqrt{1-\varepsilon^4}}(A\vert B)_\rho~,
\end{align}	
for $\rho\in\D(\H_{AB})$, and $0\leq\varepsilon \leq 1$. Substituting \eqref{ineq:from-pretty-converse} into \eqref{ineq: 1}, we find that
\begin{align}
\log_2 M &\leq  H_{\min}^{\delta'+2\sqrt{2\delta}+\delta''}(MR|\E)_\omega - H_{\max}^{\sqrt{2\varepsilon}+2\varepsilon'+\varepsilon''}(MR|\B)_\omega -\left[H_{\max}^{\sqrt{1-\delta'^4}}(R|\E)_\omega-H_{\max}^{\varepsilon'}(R|\B)_\omega\right] \nonumber \\
&\qquad +g(\delta'')+g(\varepsilon'')~.
\end{align}

We now fix $\varepsilon'= \sqrt{1-\delta'^4}$. By using the data-processing inequality of smoothed-max-conditional entropy (see, e.g., \cite{tomamichel2012framework}) under the action of a degrading channel $\T_{\B \to \E}$, we get
\begin{align}
H_{\max}^{\varepsilon'}(R|\E)_\omega \geq H_{\max}^{\varepsilon'}(R|\B)_\omega~.
\end{align}

Therefore, we can discard the terms inside of the square bracket in \eqref{ineq: 1}, and by choosing $\sqrt{1-\delta'^4}=\varepsilon'=\sqrt{2\varepsilon}$, $ \varepsilon''=2\sqrt{2\delta}$, $ \delta''=2\sqrt{2\varepsilon}$, we find
	\begin{align}
	\log_2 M &\leq H_{\min}^{3\sqrt{2\varepsilon}+2\sqrt{2\delta}}(MR|\E)_\omega - H_{\max}^{3\sqrt{2\varepsilon}+2\sqrt{2\delta}}(MR|\B)_\omega +f(\varepsilon,\delta)~,\\
	&\leq \sup_{\rho_{R'A}}\left\lbrace H_{\min}^{3\sqrt{2\varepsilon}+2\sqrt{2\delta}}(R'|\E)_\omega - H_{\max}^{3\sqrt{2\varepsilon}+2\sqrt{2\delta}}(R'|\B)_\omega\right\rbrace +f(\varepsilon,\delta)~,
	\end{align}
	where $f(\varepsilon) =- \log_2(1-\sqrt{1-8\varepsilon})(1-\sqrt{1-8\delta})$, and thus we need to impose the constraints $\delta,\varepsilon<\frac{1}{8}$ and $ 3\sqrt{2\varepsilon}+2\sqrt{2\delta}<1$. The last step follows since systems $M$ and $R$ extend the input of channel $\N$. Since these inequalities hold for any value of $\log_2M$, the desired result in \eqref{eq:one-shot-upper-bound-degraded-broadcast} follows. 
\end{proof}

\begin{remark}
The optimization in Theorem~\ref{thm:one-shot-upper-bound-degraded} is with respect to mixed-state inputs with a potentially unbounded reference system $R$. This could be viewed as undesirable. To get around this problem, we consider a purifying system $R'$ for the input state $\rho_{RA}$, iterate once more with the chain rules in \eqref{eq:chain-1} and \eqref{eq:chain-2}, and arrive at the following upper bound:
\begin{equation}
	C_{\operatorname{EP}}^{\varepsilon,\delta}(\N) \leq \max_{\psi_{RA}}\big[ H_{\min}^{3\sqrt{2\varepsilon'}+2\sqrt{2\delta'}}(R|\E)_\omega- H_{\max}^{3\sqrt{2\varepsilon'}+2\sqrt{2\delta'}}(R|\B)_\omega\big]+ f(\varepsilon',\delta')+ f(\varepsilon,\delta)~,
\end{equation}
	where the optimization is over all pure bipartite states $\psi_{RA}$ and $\varepsilon' = \delta' \equiv 3\sqrt{2\varepsilon}+2\sqrt{2\delta}$. (Note that, in the above expression, we have consolidated the systems $RR'$ external to the channel as a single system $R$.) 
The bound above is not as tight as that stated in Theorem~\ref{thm:one-shot-upper-bound-degraded}, but it has the advantage that the reference system $R$ need not be any larger than the channel input system $A$, due to the Schmidt decomposition theorem.
\end{remark}

\subsection{Asymptotic analysis}

In this section, we first define the EA private information of a quantum broadcast channel $\N_{A\rightarrow\B\cup\E}$. We then show that the EA private information of a channel $\N_{A\rightarrow\B\cup\E}$ is additive if the channel is degraded. Finally, we prove that the EA private capacity of a degraded broadcast channel $\N_{A\rightarrow\B\cup\E}$ is given by the EA private information of the channel. 

We define the EA private information of a quantum broadcast channel $\N_{A\rightarrow\B\cup\E}$ as
\begin{align}
\P_{EA}(\N)\equiv\sup_{\rho_{RA}}[ I(R;\B)_\omega - I(R;\E)_\omega]~, \label{eq:EA-private-information}
\end{align}
where $\rho_{RA}$ is an arbitrary quantum state and $\omega_{R\B\cup\E} = \N_{A\rightarrow\B\cup\E}(\rho_{RA})$. 

We now show that it is sufficient to maximize \eqref{eq:EA-private-information} with respect to only pure states if the broadcast channel $\N_{A\rightarrow \B\cup\E}$ is degraded. Consider an arbitrary input state $\rho_{RA}$ and the following chain of inequalities:
\begin{align}
 & I(R;\B)_\omega - I(R;\E)_\omega \notag \\
 &=
 I(RA'; \B)_\sigma - I(A';\B R)_\sigma 
 + I(R;A')_\sigma - [I (RA';\E)_\sigma -I(A';\E R)_\sigma + I(R;A')_\sigma]
 \\
 & = I(RA'; \B)_\sigma - I (RA';\E)_\sigma - [I(A';\B R)_\sigma-I(A';\E R)_\sigma]~\\
&\leq I(RA'; \B)_\sigma - I (RA';\E)_\sigma ~\\
& \leq \max_{\phi_{RA}} \, [I(R;\B)_{\omega} - I(R;\E)_{\omega}],
\end{align}
where $\phi_{RAA'}$ is a purification of the input state $\rho_{RA}$, and $\sigma= \N_{A\rightarrow \B\cup\E}(\phi_{RAA'})$. The first equality follows from the identity $I(A;C) = I(A;BC) - I(AC;B) + I(B;C) $. The first inequality follows from the fact that quantum mutual information decreases under the action of a degrading channel from $\B$ to $\E$. The last inequality follows from a simple relabeling of variables, and due to maximization over all input pure states. Since the chain of inequalities is true for all input states $\rho_{RA}$, the EA private information of a degraded broadcast channel is given by
\begin{align}
\P_{EA}(\N) = \max_{\phi_{RA}} \, [ I(R;\B)_\omega - I(R;\E)_\omega]~, \label{eq:EA-private-information-degraded}
\end{align}
where $\phi_{RA}$ is an arbitrary pure quantum state, and $\omega_{R\B\cup\E} = \N_{A\rightarrow\B\cup\E}(\phi_{RA})$. 

In the following lemma we prove that the EA private information of a degraded quantum broadcast channel is additive. 
\begin{lemma}\label{lemma: additivity}
	Let $\N_{A_1\rightarrow\B_1\cup\E_1}$ and $\M_{A_2\rightarrow\B_2\cup\E_2}$ be degraded quantum broadcast channels. Then
	\begin{align}
	\P_{EA}(\N \otimes \M) = \P_{EA}(\N) + \P_{EA}(\M)~,
	\end{align}
	where the EA private information of the channel $\P_{EA}(\N)$ is defined as \eqref{eq:EA-private-information}~.
\end{lemma}

	\begin{proof}
		We first prove the trivial inequality $\P_{EA}(\M) + \P_{EA}(\N) \leq \P_{EA}(\N\otimes \M)$, which holds for arbitrary quantum broadcast channels. Let $\rho_{R_1A_1}$, and $\sigma_{R_2A_2}$ be arbitrary input states. Then the following chain of inequalities holds:
		\begin{align}
		& I(R_1;\B_1)_{\N(\rho)} - I (R_1;\E_1)_{\N(\rho)} + I(R_2;\B_2)_{\M(\sigma)} - I (R_2;\E_2)_{\M(\sigma)} \notag \\
		& = I(R_1R_2;\B_1\B_2)_{(\N\otimes \M)(\rho\otimes \sigma)} - I(R_1R_2;\E_1\E_2)_{(\N\otimes \M)(\rho\otimes \sigma)}\\
		&\leq \P_{EA}(\N\otimes \M). 
		\end{align}
		The first equality follows from the definition of $\P_{EA}(\N)$. The second equality follows from additivity of mutual information with respect to tensor-product states. The final inequality follows because the input state $\rho_{R_1A_1}\otimes \sigma_{R_2A_2}$ is a particular state of the more general form $\rho_{RA_1A_2}$ needed in the optimization of the EA private information of the tensor-product channel $\N \otimes \M$. Since the inequality holds for all input states, we conclude that
		\begin{equation}
		\P_{EA}(\M) + \P_{EA}(\N) \leq \P_{EA}(\N\otimes \M).
		\end{equation}
		
		We now prove the non-trivial inequality $\P_{EA}(\N\otimes \M) \leq 
		\P_{EA}(\M) + \P_{EA}(\N) $ for degradable broadcast channels. First note that the tensor-product channel $\N\otimes \M$ is degradable because the channels individually are. So the equality in \eqref{eq:EA-private-information-degraded} applies. Let $\phi_{RA_1A_2}$ be a state that maximizes $\P_{EA}(\N\otimes \M)$, and let $\omega_{R\B_1\E_1\B_2\E_2} = (\N_{A_1\to \B_1\E_1}\otimes\M_{A_2\to \B_2\E_2})(\phi_{RA_1 A_2})$. Consider the following chain of inequalities:
		\begin{align}
		\P_{EA}(\N\otimes \M)&= I(R;\B_1\B_2)_\omega-I(R;\E_1\E_2)_\omega\\
		&= I(R;\B_1)_\omega+I(R;\B_2\vert\B_1)_\omega-I(R;\E_2)_\omega -I(R;\E_1\vert\E_2)_{\omega}\\
		&=  I(R;\B_1\vert\E_2)_\omega+I(R;\B_2\vert \B_1)_\omega - I(R;\E_2\vert \B_1)_\omega - I(R;\E_1\vert \E_2)_\omega\\ 
		& = I(R\E_2;\B_1)_\omega - I(R\E_2;\E_1) + I(R\B_1;\B_2) - I(R\B_1;\E_2) \notag \\
		& \qquad - [I(\B_1;\B_2)- I(\E_1;\E_2)]\\
		&\leq \P_{EA}(\N) +\P_{EA}(\M)~.
		\end{align}
	The second equality follows from an application of the chain rule for quantum mutual information. The third equality follows from the identity $I(A;B\vert C) = I(A;BC) - I(A;C)$. The fourth equality follows from another expression for the conditional quantum mutual information  $I(A;B\vert C) = I(AC;B) - I(B;C)$, and from a simple rearrangement. The last inequality follows from the definition of EA private information of channels $\N$ and $\M$, and from the fact that quanutm mutual information decreases under the action of a tensor product of two degrading channels.  
			\end{proof}

We now prove that the EA private capacity of a degraded quantum broadcast channel $\N_{A\rightarrow \B\cup\E}$ is given by the EA private information of the channel in the following theorem.  
\begin{theorem}\label{thm:iid_capacity}
	Let $\N_{A\rightarrow \B\cup\E}$ be a degraded quantum broadcast channel. Then the EA private capacity $C_{\operatorname{EP}} (\N)$ of the channel $\N_{A\rightarrow \B\cup\E}$ is given by
	\begin{align} 
	C_{\operatorname{EP}} (\N)= \P_{EA}(\N),
	\end{align}
	where EA private information $\P_{EA}(\N)$ is defined as \eqref{eq:EA-private-information-degraded}~.
\end{theorem}
\begin{proof}
	The direct part follows immediately from the one-shot lower bound established in Theorem~\ref{thm: one-shot-lower-bound} and the subsequent second-order expansion discussed after it in \eqref{eq:2nd-order-lower-bnd}.

	For the converse part, we begin with our one-shot upper bound established in Theorem~\ref{thm:one-shot-upper-bound}. Let $\sigma_{MR\B^n\cup \E^n} = \N^{\otimes n}_{A^n \to \B^n \cup \E^n}(\rho_{MRA^n})$.  Consider the following chain of inequalities: 
	\begin{align}
	C^{\varepsilon, \delta}_{\operatorname{EP}}(\N^{\otimes n}) &\leq \sup_{\rho_{MRA^n}}\big[ H_{\min}^{\sqrt{2\delta}}(M\vert R\E^n)_{\sigma}- H_{\max}^{\sqrt{2\varepsilon}}(M \vert R\B^n)_{\sigma}\big]\\ 
	&\leq \sup_{\rho_{MRA^n}}\big[H(M \vert R\E^n)_{\sigma} - H(M \vert RB^n )_{\sigma} \big] + f(\varepsilon, \delta, M)\\
	& =  \sup_{\rho_{MRA^n}} \big[I(M;R \B^n)_{\sigma} - I(M; R\E^n)_{\sigma}\big]+ f(\varepsilon, \delta, M)\\
	& =  \sup_{\rho_{MRA^n}} \big[I(MR; \B^n)_{\sigma} - I(M R;\E^n)_{\sigma} -[I(R; \B^n)_{\sigma} - I(R; \E^n)_{\sigma}] \big] + f(\varepsilon, \delta, M)\\
	&\leq  \sup_{\rho_{RA^n}} \big[I(R; \B^n) - I(R;\E^n)\big]_{\sigma} + f(\varepsilon, \delta, M)~,\label{eq:upper-bound-supremum-mixed-states}
	\end{align}
	where $f(\varepsilon, \delta, M) = 8 (\sqrt{2 \delta}+\sqrt{2\varepsilon}) \log_2 M +2[h_2(\sqrt{8\delta})+h_2(2\sqrt{8\varepsilon})) ]$.
	The first inequality follows by an application of Theorem \ref{thm:one-shot-upper-bound} to the tensor-power channel  $\N^{\otimes n}$. The second inequality is a consequence of  the following inequalities \cite{renes2011noisy}:
	\begin{align}
	H_{\min}^\varepsilon(A|B)_\rho\leq H(A\vert B)_\rho + 8\varepsilon\log_2\dim (A) + 2h_2(2\varepsilon)~,\label{eq:smooth-min-to-VN}\\
	H_{\max}^\varepsilon(A|B)_\rho\geq H(A\vert B)_\rho - 8\varepsilon\log_2\dim (A) -2h_2(2\varepsilon)~\label{eq:smooth-max-to-VN}.
	\end{align}
	In the above, $h_2(x)=-x\log_2 x- (1-x)\log_2(1-x)$ is the binary entropy. (Note that one could obtain improved parameters in \eqref{eq:smooth-min-to-VN} and \eqref{eq:smooth-max-to-VN} by employing recent developments in \cite{Winter15}.)   The second equality follows from the chain rule for quantum mutual information. The last inequality follows because there is a degrading channel from $\B^n$ to $\E^n$, so that the quantum data-processing inequality implies that $I(R; \B^n) \geq I(R; \E^n)$. 
	
	Next, we show that it is sufficient to take a supremum over only pure quantum states in \eqref{eq:upper-bound-supremum-mixed-states}. Consider the following chain of inequalities:
	\begin{align}
& I(R;\B^n)_{\N^{\otimes n}(\rho)}  - I(R;\E^n)_{\N^{\otimes n}(\rho)}  \nonumber \\
	&=   I(FR;\B^n)_{\N^{\otimes n}(\phi)} - I(FR;\E^n)_{\N^{\otimes n}(\phi)} -[I(F;\B^nR)_{\N^{\otimes n}(\phi)} - I(F;\E^nR)_{\N^{\otimes n}(\phi)}] \\
	 & \leq I(FR;\B^n)_{\N^{\otimes n}(\phi)} - I(FR;\E^n)_{\N^{\otimes n}(\phi)} \\
	 & \leq \P_{EA}(\N^{\otimes n})\\
	 & = n \P_{EA}(\N)~,
	\end{align}
	where $\phi_{FRA^n}$ is a purification of the state $\rho_{RA^n}$. The first equality follows from the chain rule for quantum mutual information. The first inequality follows from the fact that quantum mutual information decreases under the action of a quantum channel (particularly, a degrading channel from $\B^n$ to $\E^n$). The second equality follows from the definition of EA private information of a quantum channel, as defined in \eqref{eq:EA-private-information-degraded}. The last equality follows from Lemma \ref{lemma: additivity}. Since the chain of inequalities is true for all input states $\rho_{R A^n}$, we conclude that
	\begin{align}
	C^{\varepsilon, \delta}_{\operatorname{EP}}(\N^{\otimes n}) \leq  n \P_{EA}(\N) +  f(\varepsilon,\delta, M)~.
	\end{align}
	We now divide both sides of the last inequality by $n$ and take the limits $n\rightarrow\infty$, and $\varepsilon,\delta\rightarrow 0$:
	\begin{align}
	C_{\operatorname{EP}}(\N) & \leq \lim_{\varepsilon,\delta\rightarrow 0}\liminf_{n\rightarrow\infty}\big[\P_{EA}(\N) +\frac{1}{n}f(\varepsilon,\delta, M)]\\
	& = \P_{EA}(\N)~.	\label{eq:CER_upper}
	\end{align}
	Hence, the lower bound in \eqref{eq:2nd-order-lower-bnd} and the upper bound in \eqref{eq:CER_upper} imply that
	\begin{align}
	C_{EP}(\N) = \P_{EA}(\N)~.
	\end{align}
	This concludes the proof.
	\end{proof}

\section{Examples of degraded quantum broadcast channels}
\label{sec:ex}

In this section, we consider EA private communication over two specific instances of degraded quantum broadcast channels. Moreover, we establish an operational meaning for the conditional quantum mutual information of a quantum broadcast channel, as a dynamic counterpart of the prior result from \cite{sharma2017conditional}.

\subsection{When Eve has access to the pre-shared entanglement and some part of Bob's laboratory}

\begin{figure}
	\centering
	\includegraphics[scale=0.7]{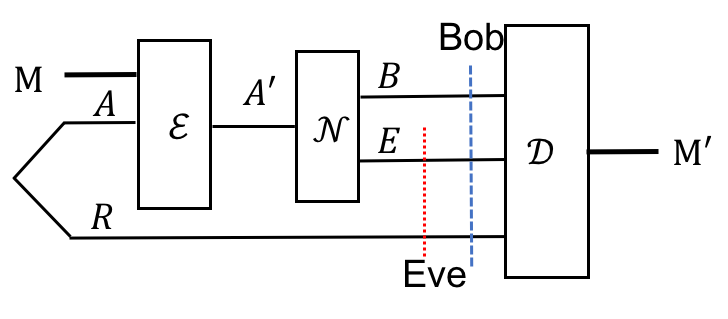}
	\caption{The information-processing task for EA private communication via a two-receiver broadcast channel. Bob, who has access to the systems $R$, $B$ and $E$, can reliably decode the transmitted message. However, Eve, who has access systems $R$ and $E$, can only get negligible information about Alice's message. }\label{fig:case-1}
\end{figure}

We consider a special case of the protocol considered in Section ~\ref{sec:capacity-degrade}, as shown in Figure~\ref{fig:case-1}.
Alice would like to transmit a classical message $m$ from a set $\mathcal{M} \equiv \{1,...,M\}$, to Bob via a quantum channel $\N_{A\rightarrow BE}$, where Bob's lab is separated into two parts. Bob has access to both parts, but we assume that the second part of the Bob's lab ($E$) is insecure---a malicious third party Eve has access to this second part. Moreover, Alice and Bob pre-share arbitrary entanglement to assist their communication. The quantum channel $\N_{A\rightarrow BE}$ is a special case of previously defined degraded quantum broadcast channel, where we identify $\B=BE$ and $\E=E$. Hence, the degradable channel from $\B$ to $\E$ is simply the partial trace over the $B$ system.

We define an $(M,\varepsilon, \delta)$ code to be a set of encoding channels and a decoding POVM $\brac{\E^m_{A\rightarrow A'}, \Lambda^m_{RBE}}_m$, such that

\begin{enumerate}
\item the classical messages can be reliably decoded by Bob,
\begin{align}
\max_{m\in\M}p_e(m)\leq \varepsilon~,
\end{align}
where $
p_e(m) = \tr\brac{(I-\Lambda^m_{RBE})\rho^m_{RBE}}$, 
$\rho^m_{RBE} \equiv\N_{A'\rightarrow BE}(\E_{A\rightarrow A'}(\Psi_{RA}))
$, and
\item each classical message is $\delta$-secure:
\begin{align}
\frac{1}{2}\Vert \rho^m_{RE}-\sigma_{RE}\Vert_1\leq\delta,~~~\forall m\in\M~,
\end{align}
where $\sigma_{RE}$ is some constant state. 
\end{enumerate}

\begin{corollary} \label{thm:i.i.d.}
	Let $\N_{A\rightarrow BE}$ be an arbitrary quantum broadcast channel. Then the EA private capacity for the scenario discussed above is given by conditional quantum mutual information of the channel $\N_{A\rightarrow BE}$,
	\begin{align}
	C_{\operatorname{EP}}(\N) = \operatorname{CMI}(\N)~,
	\end{align}
	where $\operatorname{CMI}(\N)$ is defined as 
	\begin{align}
	\operatorname{CMI}(\N) =\max_{\phi_{RA}}I(R;B|E)_\omega~,
	\end{align}
	$\phi_{RA}$ is a pure bipartite state, and $\omega_{RBE}=\N_{A\rightarrow BE}(\phi_{RA})$.
\end{corollary}

\begin{proof}
	 Let $\B=BE$ and $\E=E$. Then $\N_{A\rightarrow \B \E}$ is a degraded quantum broadcast channel, such that the partial trace over system $B$ is the degrading channel from $\B$ to $\E$. Using this in Theorem~\ref{thm:iid_capacity}, we get
     \begin{align}
     C_{\operatorname{EP}}(\N) &= \max_{\phi_{RA}}[I(R;BE)_\omega - I(R;E)_\omega]
     	=  \max_{\phi_{RA}}I(R;B|E)_\omega~,
     \end{align}
where the last equality follows from the definition of conditional quantum mutual information.
%
%
%
\end{proof}

\subsection{When Eve has access to the pre-shared entanglement and no access to Bob's laboratory}
\label{sec:dchannel}

\begin{figure}
	\centering
	\includegraphics[scale=0.7]{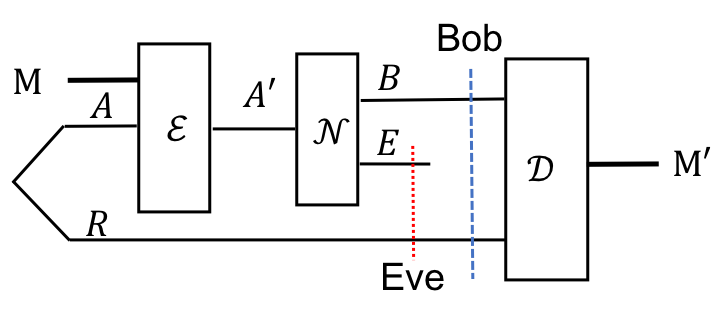}
	\caption{The information-processing task for EA private communication through a quantum channel $\N$. Bob, who has access to the systems $R$ and $B$, can reliably decode the transmitted message. However, Eve, who has access systems $R$ and $E$, can only get negligible information about Alice's message.  }\label{fig:case-2}
\end{figure}

In this section, we consider another special case of the protocol considered in Section~\ref{sec:capacity-degrade}, as shown in Figure~\ref{fig:case-2}. Similarly, Alice would like to securely transmit a classical message $m$, chosen from a set $\M \equiv \{1,...,M\}$, to Bob via a quantum channel $\N_{A\rightarrow BE}$. We assume that there is an eavesdropper who can access systems $R$ and $E$. We also suppose now that there exists a degrading channel from $B$ to $E$.

 We now define an $(M, \varepsilon, \delta)$ code to be a set  of encoding channels and a decoding POVM $\{\E^m_{A}, \Lambda^m_{RB} \}_m$, such that 
\begin{enumerate}
\item the classical messages can be reliably decoded by Bob,
\begin{align}\label{ineq:reliable-2}
\max_{m\in\M}p_e(m)\leq \varepsilon~,
\end{align}
where $
p_e(m) = \tr\brac{(I-\Lambda^m_{RB})\rho^m_{RB}}$ and $
 \rho^m_{RBE} \equiv \N_{A'\rightarrow BE} (\E^m_{A \rightarrow A'}(\Psi_{RA}))$,
\item each classical message is $\delta$-secure:
\begin{align}\label{ineq:security-2}
\frac{1}{2}\Vert \rho^m_{RE}-\sigma_{RE}\Vert_1\leq\delta,~~~\forall m\in\M~,
\end{align}
where $\sigma_{RE}$ is some constant state. 
\end{enumerate}

\begin{corollary}
	Let $\N_{A\rightarrow BE}$ be a degraded quantum broadcast channel as defined above. Then the EA private capacity is given by
	\begin{align}
	\max_{\phi_{RA}} [I(R;B)_{\omega}-I(R;E)_{\omega}]~,
	\end{align}
	where $\phi_{RA}$ is any pure quantum state and $\omega_{RBE}=\N_{A\rightarrow BE}(\phi_{RA})$.
\end{corollary}
\begin{proof}
In this case, $\B = B$ and $\E = E$. Applying Theorem~\ref{thm:iid_capacity}, we find that 
$
C_{EA}(\N) = \max_{\phi_{RA}}[I(R;B)_\omega - I(R;E)_\omega]~.
$
\end{proof}

\section{Conclusion} \label{sec:conclusion}
In this work, we considered EA private communication over quantum broadcast quantum channels.  We established a lower bound on the one-shot EA capacity based on the recent techniques of position-based coding \cite{anshu2017one} and convex splitting \cite{anshu2014quantum}. We also established an upper bound on the one-shot EA private capacity by combining various results on the min- and max-entropy. We defined a quantum broadcast channel to be degraded when there is a quantum channel mapping from Bob's systems to Eve's systems. Using lower and upper bounds on the one-shot EA private capacity, we proved a single-letter EA capacity formula for degraded quantum broadcast channels. As special cases, we found capacities of EA private communication over two-receiver degraded broadcast channels. Especially in the first case, we not only proved a single-letter capacity formula, but also established an operational meaning to conditional quantum mutual information of a quantum broadcast channel. One possible future direction is to investigate whether our upper bound is also a strong converse upper bound.  Another intriguing idea is to generalize our results to the scenario of secret sharing.

\bigskip
\textbf{Acknowledgements.}
We are grateful to Anurag Anshu, Rahul Jain, Felix Leditzky, and Naqueeb Warsi for discussions related to the topic of this paper.
All authors acknowledge support from the Department of Physics and Astronomy at LSU and the National Science Foundation under Grant No.~1714215.

\bibliographystyle{alpha}
\bibliography{Ref}
\end{document}